\documentclass{svjour3}             % onecolumn (standard format)
\smartqed  % flush right qed marks, e.g. at end of proof
\usepackage{graphicx}
\usepackage{amsmath}
\usepackage{amssymb}
\usepackage{filecontents}
\begin{document}

\title{From $\mu$\mbox{-}Calculus to Alternating Tree Automata using Parity Games}


\author{M. Fareed Arif}


\institute{F. Author \at
	     CASL, University College Dublin, Ireland,\\
             Tel.:+353 (1) 7162641 \\
             Fax: +353 (1) 7165396\\
             \email{12250832@ucdconnect.ie}
%           \emph{Present address:} of F. Author  %  if needed
%          \and
%          S. Author \at
%             second address
}

\date{\today}
\maketitle

\begin{abstract}
$\mu$\mbox{-}Calculus and automata on infinite trees are complementary ways of describing infinite tree languages.
The correspondence between $\mu$-Calculus and alternating tree automaton is used to solve the satisfiability and model checking problems by compiling the modal $\mu$\mbox{-}Calculus formula into an alternating tree automata. Thus advocating an automaton model specially tailored for working with modal $\mu$\mbox{-}Calculus. The advantage of the automaton model is its ability to deal with arbitrary branching in a much simpler way as compare to the one proposed by Janin and Walukiewicz. Both problems (i.e., model checking and satisfiability) are solved by reduction to the corresponding problems of alternating tree automata, namely to the acceptance and the non\mbox{-}emptiness problems,  respectively. These problems, in turn, are solved using parity games where semantics of alternating tree automata is translated to a winning strategy in an appropriate parity game.
\end{abstract}
\section{Introduction}
\label{intro}

There is a growing recognition to apply formal mathematical methods for specifying and reasoning about the correctness of computing systems namely in model checking. Model checking presents an efficient and more expressive method to automatically verify a transition system whether it meets a correctness specification formulated in $\mu$\mbox{-}Calculus. Modal $\mu$\mbox{-}Calculus is a fundamental logic for specifying properties of transition systems. It is a quite expressive language that subsumes most of the common logical formalisms used in verification including LTL, CTL, CTL* and PDL~\cite{Manol00}. The $\mu$\mbox{-}Calculus over binary trees coincide in expressive power with alternating tree automata and it is as expressive as Monadic Second-order logic (MSOL) on trees~\cite{Igor02}. 

In this paper, we explain a method to translate a given $\mu$\mbox{-}Calculus formula $\varphi$ into an appropriate finite automaton $\mathcal{A}$ such that $L(\varphi) = L(\mathcal{A})$~\cite{Wilke00}. Such a translation redcues the model checking and the satisfiability problem in $\mu$\mbox{-}Calculus to the word and the emptiness problem for finite automata~\cite{Daniel02}. Thus resulting an algorithm to solve the model checking and the satisfiability problem on a computer. Parity game play a crucial role within this translation, since it defines the semantics of alternating tree automata (i.e., whether an automaton accepts or rejects some transition system by the existsence of a winning strategy for a player in an appropriate parity game)~\cite{Zappe02}. Parity games provide a straight froward, convent construction to complement a given alternating tree automaton and also use to show the decidability of the word problem and the emptiness problem.

This paper proceeds as in the Sec.2 transitions systems are explained along with alternating tree automaton. Sec.3 comprises of parity games and automaton acceptance. Sec.4 introduce $\mu$\mbox{-}Calculus formalism where in Sec.5 we define encoding from $\mu$\mbox{-}Calculus to alternating tree automaton, concluding the last section by explaining the satisfiability and model checking problems under defined encoding along with complexity bounds. %Appendix discusses proves.

\section{Transition Systems and Alternating Tree Automaton}
\textit{Transition Systems} are structures consisting of a non\mbox{-}empty set of states, a set of unary relation describing properties of states and a set of binary relation describing transitions between states. In short, transition systems are describing the operational semantics of any program. Thus model checking amounts to verifying that any corresponding transition system has a property of interest. Properties of interest are formulated in logic languages like (LTL, CTL, CTL*, PDL, $\mu$\mbox{-}Calculus and MSOL);
\begin{definition}\textit{(Transition System)}
A \textbf{Transition system} (i.e., Kripke Structure) over $\mathcal{P}$ is a tripe $\mathcal{K} = (S, R, \lambda)$ where
\begin{itemize}
\item  $\mathcal{P}$ be a set of atomic propositions (properties) and for any propositional interpretation $\mathcal{I}: \mathcal{P} \rightarrow \{true,false\}$.
 \item S is a set called \textbf{states} (worlds), universe of $\mathcal{K}$,
\item $R \subseteq S \times S$ is a transition relation and
\item $\lambda: S \rightarrow 2^{\mathcal{P}}$ is a mapping (i.e., $\lambda(s_i) = p_i$ for every $p_i \in  \mathcal{P}$). $\lambda(s_i) = p_i$  if  $p_i$ is true in $s_i$ and $\neg p_i$ if $p_i$ is false in $s_i$.
\end{itemize}
$\lambda: S \rightarrow 2^{\mathcal{P}}$ regards transition systems as labeled directed graphs. For every $s \in S$, we denote
\begin{equation}
 sR = \{s' \in S| (s,s')\in R\}, Rs = \{s'\in S| (s',s)\in R \}
\label{eqno1}
\end{equation}
\end{definition}

A \textbf{pointed transition system} (i.e, a rooted kripke structure) is a pair $(\mathcal{K},s_I)$ in a transition system $\mathcal{K} = (S, R, \lambda)$ with an initial state $s_I \in S$~\cite{Daniel02}.

An alternating tree automata is a device which accepts or rejects pointed transition systems by parsing the paths.
\begin{definition}\textit{(Alternating Tree Automata)}
 An alternating tree is a tuple $\mathcal{A} = \{Q, q_I,\delta,\Omega\}$ where
\begin{itemize}
 \item $Q$ is a finite set of states of the automaton,
 \item $q_I \in Q$ is a state called the initial state,
\item $\delta: Q \rightarrow TC^Q$ is a transition function which maps every state $q\in Q$ to a transition condition $TQ$ where all the transition conditions $TQ$ over $Q$ are defined by:
\begin{itemize}
\item $0$ and $1$ are transition conditions over $Q$.
\item $p, \neg p$ are transition conditions over $Q$, for every $p\in \mathcal{P}$.
\item $q, \Box q, \Diamond q$ are transition conditions over $Q$, for every $q \in Q$.
\item $q_1 \wedge q_2, q_1 \vee q_2$ are transition conditions over $Q$, for every $q_1,q_2\in Q$.
\end{itemize}
\item $\Omega: Q \rightarrow \omega$ is called priority function (coloring function) which assigns color to states of $\mathcal{A}$.
\end{itemize} 
\end{definition}

\subsection{Index of Alternating Tree Automata}
Concerning the complexity of an alternating tree automata, an important notion is its index~\cite{Wilke00}. Let $\mathcal{A} = \{Q, q_I,\delta,\Omega\}$ be an alternating tree automata, then the transition graph $G(\mathcal{A})$ has the set $Q$ as vertex set. There is a edge relation from a vertex $q$ to $q'$ iff $q'$ appears in the transition condition $\delta(q)$.

Let $\mathcal{C}^{\mathcal{A}}$ be the set of all strongly connected components of the transition graph $G(\mathcal{A})$ of $\mathcal{A}$. For every $C \in \mathcal{C}^\mathcal{A}$, let
\begin{equation}
 m_C^{\mathcal{A}} = |\{\Omega^{\mathcal{A}}(q)|q\in C\}|
\end{equation}
denote the number of priorities used in $C$. The index of $\mathcal{A}$, denoted by $index(\mathcal{A})$, is the maximum of all these values, that is,
\begin{equation}
index(\mathcal{A}) = max(\{m_C^{\mathcal{A}}|C\in \mathcal{C}^{\mathcal{A}}\}\cup\{0\})
\end{equation}

\section{Infinite Games \& Parity Acceptance Conditions}
This section introduce infinite two-person games on directed graphs along with winning play strategies for a certain player~\cite{Rene02}. 
%Some fundamentals notions such as determinacy, forgetful strategies and memoryless strategies are discussed;
\begin{definition}\textit{(Game)}
\label{def:game}
\end{definition}
A game is composed of an arena and a winding condition. Let $\mathcal{A}$ be an arena then the pair $\mathcal{G} = (\mathcal{A}, Win)$ is called a $\textbf{game}$ where $Win\subseteq V^{\omega}$ is a winning set where $\omega$ is infinite supply of intergers. 
\begin{definition}\textit{(Arena)}
\label{def:arena}
An \textbf{arena} is a triple
\begin{equation}
 \mathcal{A} = (V_0, V_1, E) 
\end{equation}
where $V_0$ is a set of $0$\mbox{-}vertices, $V_1$ a set of $1$\mbox{-}vertices and $E \subseteq (V_0 \cup V_1)\times (V_0 \cup V_1)$ is the $\textbf{edge relation}$ also known as set of moves. The union $V = (V_0 \cup V_1)$ where $V_0$ and $V_1$ are disjoint sets of vertices. Under this union requirement edge relation correspond to $E \subseteq V \times V$. The set of $\textbf{successor}$ of $v\in V$ is defined as $vE = \{v' \in V | (v,v')\in E\}$
%Two players, Play 0 and Player 1 of interest in these games. 
Consider a player $\sigma$ for $\sigma \in \{0,1\}$, then the opponent of player $\sigma$ is player $\bar{\sigma}$ (i.e., $\bar{\sigma} =  1 - \sigma$).
\end{definition}

\begin{definition}\textit(Play)
 \label{def:play}
We define a play in the arena $\mathcal{A}$ as followed:
\begin{itemize}
 \item a finite play $\pi = v_0 v_1\dots v_l \in V^+$ with $ v_{i+1} \in v_i E$ for all $i < l$ and $v_l E = \emptyset$ represents a dead\mbox{-}end, a prefix of this finite play is $\rho(\pi) = v_0v_1\dots v_k$  for $k \leq l$.
\item an infinite play $\pi = v_0 v_1\dots v_l \in V^\omega$ with $v_{i+1} \in v_i E$ for all i $\in \omega$, a prefix for this infinite play is $\rho(\pi) = v_0v_1\dots v_k$ for $k \geq 0$
\end{itemize}
\end{definition}

\begin{definition}\textit{(Winning Set)}
 \label{def:win}
To define the winning conditions for Players (Player 0, Player 1) are as followed:
 \end{definition}

Player $\sigma$ is declared the winner of a play $\pi$ in the game $\mathcal{G}$ iff
\begin{itemize}
 \item $\pi$ is a finite play $\pi = v_0 v_1\dots v_l \in V^+$ and $v_l$ is a $\bar{\sigma}$\mbox{-}Vertex where Player $\bar{\sigma}$ can not move anymore (i.e., $v_l$ is a dead\mbox{-}end, $v_lE = \emptyset$) or
\item $\pi$ is an infinite play and $\pi \in Win$.
\end{itemize}
Conversely, Player $\bar{\sigma}$ wins play $\pi$ if Player $\sigma$ does not win $\pi$. 

In every play, a token is placed on some initial vertex $v \in V$. If $v$ is $\sigma$\mbox{-}vertex then Player $\sigma$ moves the token from $v$ to $v' \in vE$, symmetrically moves for Player $\bar{\sigma}$ are considered in case of $\bar{\sigma}$\mbox{-}vertex. This process is repeated infinity often or until a dead end is reached (i.e., a vertex without successor).

\begin{definition}\textit{(Coloring Function)}
 The coloring function $\chi: V \rightarrow C$ color vertices of arena $\mathcal{A}$ where $C$ is a finite set of colors (priorities)(i.e, $C \subseteq \mathbb{N}$) and it extends to an infinite play $\pi=v_0v_1\dots$ as $\chi(\pi)=\chi(v_0)\chi(v_1)\dots$. 
\end{definition}
Let $Win$ is an acceptance condition for an automaton then $W_{\chi}(Win)$ is the winning set consisting of all infinite plays $\pi$ where $\chi(\pi)$ is accepted according to $Win$.
\begin{itemize}
 \item Parity conditions or Colour set C is a finite subset of integers and $Inf(\chi(\pi))$ be the set of colors that occurs infinitely often in $\chi(\pi)$ then for, 
  \subitem Max-parity condition: $\pi \in W_{\chi}(Win)$ iff $max(Inf(\chi(\pi)))$ is even. 
  \subitem Min-parity condition: $\pi \in W_{\chi}(Win)$ iff $min(Inf(\chi(\pi)))$ is even.
\end{itemize}

\begin{example}
%\begin{figure}[Arena]
%\centering
\includegraphics[width=0.4\textwidth]{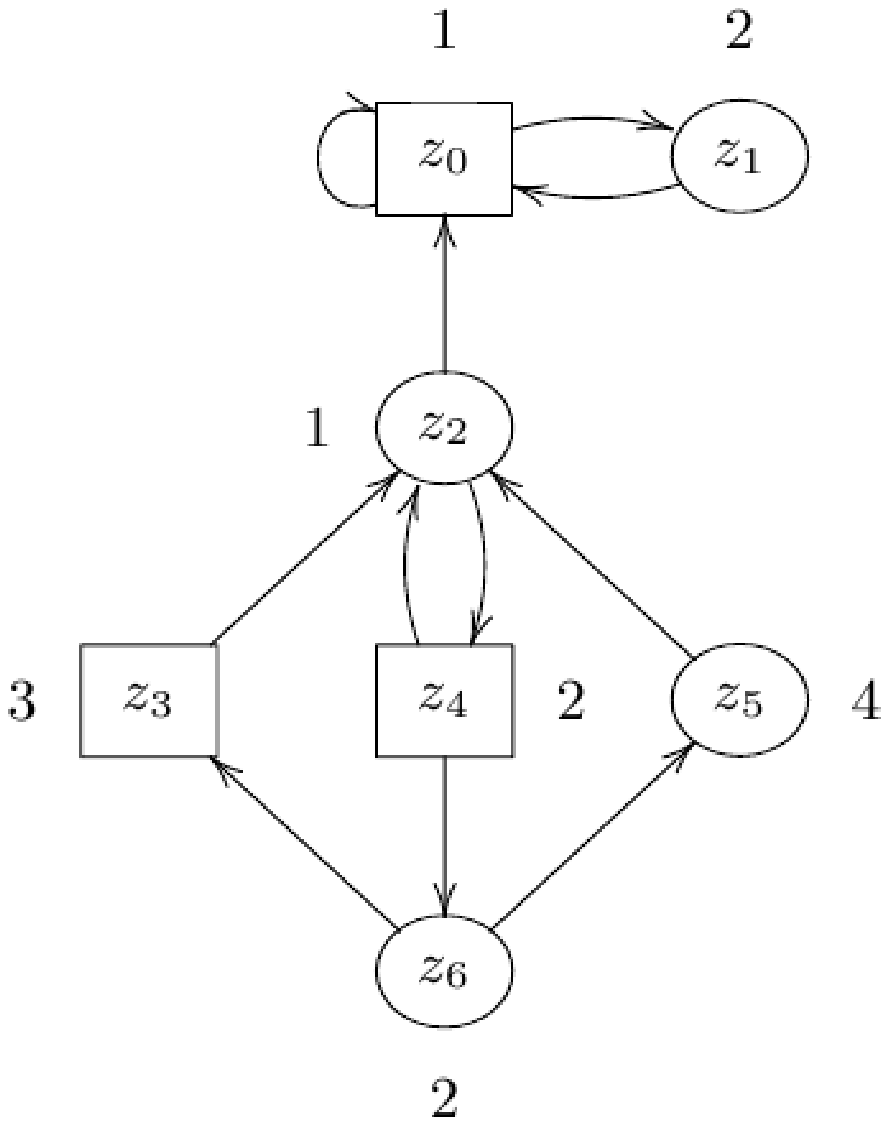}
%\caption	
%{Figure Example 2.1}\label{fig:arena}
%\end{figure}
%Example 2.1 

In above figure, let $\mathcal(G) = (\mathcal{A}, Win)$ defines an arena where $ \mathcal{A} =(V_0,V_1,E)$ such that $V_0 =\{z_1,z_2,z_5,z_6\}$ (circles), $V_1 =\{z_0,z_3,z_4\}$ (squares), Coloring set $C = \{1,2,3,4\}$ and 
$\chi(z_4) = 2$ as shows in figure; winning set of condition $Win = \{\{1,2\},\{1,2,3,4\}\}$. In a possible infinite play in this is $\pi = z_6z_3z_2z_4z_2z_4z_6zz_5(z_2z_4)^{\omega}$. According to Muller acceptance condition (i.e., $\pi \in W_{\chi}(Win)$ iff $Inf(\chi(\pi)) \in \mathcal{A}$) this play $\pi$ is winning for Player $0$ because $\chi(\pi)=23121224(12)^{\omega}$ where $Inf(\chi(\pi))=\{1,2\} \in Win$. For play $\pi'=(z_2z_4z_6z_3)^{\omega}$ yields $\chi(\pi')=(1223)^{\omega}$ and $Inf(\chi(\pi'))=\{1,2,3\} \not\in Win$, hence $\pi'$ is winning for Player $1$. Regarding parity conditions this play is a loss for player $0$ because  $min(Inf(\chi(\pi))=\{1\})$ is odd., hence a win for the opponent.
\end{example}

\subsection{Behavior of Alternating Tree Automata}
Let $(\mathcal{K},s_I)$ be a pointed transition system and let $\mathcal{A}=(Q,q_I,\delta,\Omega)$ be an alternating tree automaton. To define the behavior of $\mathcal{A}$ on $(\mathcal{K},s_I)$ consider alphabet over sequences of pairs ($Q \times S$)~\cite{Daniel02}.

For a word $w \in (Q \times S)^*$  where $(q,s)\in (Q \times S)$ is an alphabet such that the \textit{behavior of $\mathcal{A}$ on $(\mathcal{K},s_I)$} is the language $\mathcal{L(\mathcal{A},\mathcal{K})} \subseteq (Q \times S)^*$ (i.e., $w \in L(\mathcal{A},\mathcal{K})$). The initial state correspond to $(q_I,s_I) \in \mathcal{\mathcal{L(\mathcal{A},\mathcal{K})}}$. To define the transition function consider that the automaton is in the state $q$ and it inspects the state $s$ (i.e., the current instance is $(q,s)$).
Now, the automaton tries to execute the transition condition $\delta(q)$ defined as followed:
\begin{itemize}
 \item If $\delta(q)\in \{0,1\}$ or $\delta(q) = p$ or $\delta(q) = \neg p$ for some proposition $p \in \mathcal{P}$ then the automaton $\mathcal{A}$ need not to take any action.
\item If $\delta(q) = q'$ for some $q' \in Q$, then from current state updates $(q,s) \rightarrow (q',s)$, then $(q,s)(q',s) \in \mathcal{L(\mathcal{A},\mathcal{K})}$.
\item If $\delta(q) = q_1 \wedge q_2$ or $\delta{q} = q_1 \vee q_2$ for some $q \in Q$, then $\mathcal{A}$ splits itself into two instances $(q_1,s)$ and $(q_2,s)$ such that $(q_1,s)(q_2,s) \in \mathcal{L(\mathcal{A},\mathcal{K})}$. 
\item If $\delta(q) = \Box q'$ or $\delta(q) = \Diamond q'$ for some $q'\in Q$, then automanton splits in several instances such that for every $s' \in sR$, $(q,s)(q',s')\in \mathcal{L(\mathcal{A},\mathcal{K})}$.
\end{itemize}

Now the main concern is how the automaton $\mathcal{A}$ \textit{accepts or rejects} the pointed transition system $(\mathcal{K},s_I)$. 

Lets developed a notion of a \textbf{successful instance} for an instance $(q,s)$ such that:
\begin{itemize}
 \item If $\delta(q) = 1$, then the instance succeeds and if $\delta(q) = 0$, then it does not succeed. 
 \item If $\delta(q)$ is a propositional variable $p \in \mathcal{P}$ and $p$ is true in the state $s$ (i.e., $\delta(q) = p$ and $p \in \lambda(s)$) then instance is successful. Similarly, holds for $\delta(q) = \neg p$ and $p \not\in \lambda(s)$. 
 \item Converseley, if $\delta(q) = p$ and $p \not\in \lambda(s)$ or $\delta(q) = \neg p$ and $p \in \lambda(s)$ then the instance $(q,s)$ is not successful.
\item If $\delta(q) = q'$, then automaton changes its state to new instance $(q',s)$ such that the instance $(q,s)$ is successful iff $(q',s)$ is successful.
\item If $\delta(q_1 \wedge q_2)$, then the instance $(q,s)$ succeeds iff both instances $(q_1,s)$ and $(q_2,s)$ succeed.
\item If $\delta(q_1 \vee q_2)$, then the instance $(q,s)$ succeeds iff at least one of the the instance $(q_1,s)$ and $(q_2,s)$ succeeds.
\item If $\delta(q) = \Box q'$, then iinstance $(q,s)$ succeeds iff $\forall s' \in sR$ the instances $(q,s')$ succeeds.
\item Finally, if $\delta(q) = \Diamond q'$, then instance $(q,s)$ succeeds iff $\exists s' \in sR$ such that an instance $(q,s')$ succeeds.
 \end{itemize}

The result of this process is a parse tree with instances as nodes;
\begin{definition}%Proposition instead of theorem
The automaton $\mathcal{A}$ accepts the finite transition system $(\mathcal{K},s_I)$ iff the initial instance $(q_I,s_I) \in \mathcal{L(\mathcal{A},\mathcal{K})}$ succeeds.
\end{definition}

Above mentioned formalization for the notion of \textit{successful instance} encounter following problems:
\begin{itemize}
 \item  If parse tree is infinite, then successful instances cannot be determined in a bottom\mbox{-}up fashion.
\item If $\delta(q) = q'$, then the instance $(q,s)$ is successful iff $(q',s)$ is successful. Howere, if $\delta(q) = q$, then it end up in an \textit{infinite loop}.
\end{itemize}
To resolve these problems; "evaluation problem" is considered as solving a certain game in an infinite play where acceptance is decided according to a winning condition as explained in the next section of parity games.

\subsection{Automata acceptance using Parity conditions}

The automaton $\mathcal{B}=(Q,q_I,\delta,\Omega)$ \textbf{accepts} the pointed transition system $(\mathcal{K},s_I)$ iff there a winning strategy for Player $0$ in the $\mathcal{G}$ where $\mathcal{G} = (\mathcal{A}, Win)$ is a two player (i.e., Player $0$, Player $1$) parity game where $\mathcal{A} = (V_0, V_1, E)$ defines an arena, $V = (V_0 \cup V_1)$ and $Win \subseteq V^{\omega}$ is a winning set such that:
\\
A vertex $v = (q, s) \in V_0$ where $q\in Q$ and $s\in S$ iff 
\begin{itemize}
\item $\delta(q) = 0$.
\item $\delta(q) = p$, $p \not\in \lambda(s)$.
\item  $\delta(q) = \neg p$, $ p \in \lambda(s)$.
\item $\delta(q) = q'$ for some $q' \in Q$.
\item $\delta(q) = q_1 \vee q_2$ for $q_1,q_2 \in Q$
\item $\delta(q) = \Diamond q'$.
\end{itemize}
 
A vertex $v = (q, s) \in V_1$ iff 

\begin{itemize}
\item $\delta(q) = 1$.
\item $\delta(q) = p$, $p \in \lambda(s)$.
\item  $\delta(q) = \neg p$, $ p \not\in \lambda(s)$.
\item $\delta(q) = q_1 \wedge q_2$ for $q_1,q_2 \in Q$
\item $\delta(q) = \Box q'$.
\end{itemize}
In parity Game $\mathcal{G}$, the edge relation $E$ is defined as:
\begin{equation}
 E := \{((q,s),(q_1,s_1))| (q_1,s_1)\in (q,s)R \wedge (q,s),(q_1,s_1)\in V\}.
\end{equation}
For alternating tree automata $\mathcal{B}$ the priority function is defined by:
\begin{equation}
 \Omega(q,s) := \Omega(q) \texttt{ for } (q,s) \in V.
\end{equation}

%Discuss Winning Strategy; 
%
\subsection{Winning Strategy}
A strategy for a Player $\sigma$ is the function $f_\sigma: V_{\sigma}\times V \rightarrow V$ and a prefix play for $\pi = v_0v_1\dots v_l$ is conform with $f_{\sigma}$ if for every $i$ with $0\leq i <l$ and $v_i \in V_{\sigma}$ is defined and we have for $v_{i+1}=f_{\sigma}(v_0\dots v_i)$. A play is conform with $f_{\sigma}$ if each of its prefix is conform with $f_{\sigma}$. $f_{\sigma}$ is a strategy for Player $\sigma$ on $U \subseteq V$ if it is defined for every prefix of a play which is conform with it, starts in a vertex $U$ and does not end in a dead end for player $\sigma$. A strategy $f_{\sigma}$ is a winning strategy for Player $\sigma$ on $U$ if all plays which are conform with $f_{\sigma}$ and start from a vertex in $U$ are wins for $Player \sigma$.

\begin{definition}
 Player $\sigma$ wins a game $\mathcal{G}$ on $U\subseteq V$ if he has a winning strategy on $U$.
\end{definition}

The winning region for Player $\sigma$ is the set $W_{\sigma}(\mathcal{G})\subseteq V$ of all vertices such that Player $\sigma$ wins $(\mathcal{G},v)$ (i.e., Player $0$ wins $\mathcal{G}$ on $\{v\}$). Hence, for any $\mathcal{G}$, Player $\sigma$ wins $\mathcal{G}$ on $W_{\sigma}(\mathcal{G})$.

\begin{definition}
The automaton  $\mathcal{A} = \{Q, q_I,\delta,\Omega\}$ accepts the pointed transition system $(\mathcal{K},s_I)$ iff there is a winning strategy for Player $0$ in the parity game $\mathcal{G} = ((V_0, V_1, E), \Omega(q))$.
\end{definition}
%proof.

The language of automaton $\mathcal{A}$ consists of the pointed transition systems which $\mathcal{A}$ accepts and is denoted by $L(\mathcal{A})$~\cite{Daniel02}.

\begin{example}
 Let $\delta(q) = \Box q_I$ and $\Omega(q_I) = 0$. Let $(\mathcal{K},s_I)$ by any pointed transition system. Player $1$ can't win since he losses every finite and infinite play because the only priority function $\Omega(q) = 0$, which complies to single vertex $(q_I,s_I)\in V_0$  with infinite loop itself hence accepts every pointed transition system. 
\end{example}

\begin{example}
 Let $\delta(q) = \Diamond q_I$ and $\Omega(q_I) = 1$. This automaton does not accept any pointed transition system.
\end{example}

\begin{proposition}\textit{(Word Problem)}
The word problem is to decide whether a given alternating tree automaton $\mathcal{A} = \{Q, q_I,\delta,\Omega\}$ accepts a given finite pointed transition system $(\mathcal{K},s_I)$ .
\end{proposition}
\begin{proposition}\textit{(Emptiness Problem)}
 The emptiness problem is to show that an alternating tree automaton $\mathcal{A} = \{Q, q_I,\delta,\Omega\}$ accepts if $\mathcal{A}$ accepts at least one transition system.
\end{proposition}

\section{$\mu$\mbox{-}Calculus}
\subsection{Syntax and Semantics of $\mu$\mbox{-}Calculus}
This section introduce modal $\mu$\mbox{-}Calculus by presenting syntax and semantics and later we introduce the notion of Tarski Fixed point theorem;
\subsection{Syntax of $\mu$\mbox{-}Calculus}
\begin{definition}
\label{def:mu_calculus}
The set $L_{\mu}$ is a set of inductively defined modal $\mu$\mbox{-}Calculus formulas:
\begin{itemize}
\item $\bot$, $\top \in L_{\mu}$.
\item For every atomic proposition $p \in \mathcal{P}$; $p, \neg p \in L_{\mu}$.
\item If $\varphi,\psi \in L_{\mu}$, then $\varphi \circ \psi \in L_{\mu}$ where $ \circ \in \{\vee, \wedge\}$.
\item If $\varphi \in L_{\mu}$, then $\Box\varphi,\Diamond \varphi \in L_{\mu}$.
\item If $p \in \mathcal{P}$, $\varphi \in L_{\mu}$, and $p$ occurs only positively in $\varphi$ then $\mu p\varphi, \nu p\varphi \in L_{\mu}$.
\end{itemize}
\end{definition}
\subsection{Fixed point Operators \& Free variable set}
The operators $\mu$ and $\nu$ are called \textit{fixed point operators} (i.e., $\mu$, least fixed\mbox{-}point and $\nu$, greatest fixed\mbox{-}point) viewed as quantifiers ensuring that the argument of a fixed\mbox{-}point operator as a monotone function~\cite{Zappe02}. Accordingly, the set $free(\varphi)$ of free variables of an $L_{\mu}$ formula $\varphi$ is defined inductively as follows:
\begin{itemize}
 \item $free(\bot)=free(\top)= \emptyset$,
 \item $free(p) = free(\neg p)=  \{p\}$,
 \item $free(\varphi \circ \psi) = free(\varphi)\cup free(\psi)\}$ where $\circ\in\{\wedge,\vee\}$,
 \item $free(\Box \varphi) = free(\Diamond \varphi) = free(\varphi)$,
 \item $free(\mu p\varphi) = free(\nu p\varphi)=free(\varphi)/\{p\}$.
\end{itemize}

The sets $F_{\mu}$ and $F_{\nu}$ are defined as follows:
\begin{eqnarray}
 F_{\mu} = \{\mu p\psi| \psi\in L_{\mu}\},
 F_{\nu} = \{\nu p\psi| \psi\in L_{\nu}\}.
\end{eqnarray}

Formulas from the set $F_{\eta} = F_{\mu}\cup F_{\nu}$ are called fixed point formulas.

\subsection{Fixed point Alternation}
Beside its length, the most important characteristic of a formula $\varphi \in L_{\mu}$ is its fixed point alternation depth, that is, the number of alternations between least and greatest fixed point operators.
We now define the notion of alternation depth for $L_{\mu}$ formulas which coincides with to the notion of index for an alternating tree automaton (i.e., $\alpha(\varphi) = index(\mathcal{A}(\varphi))$)~\cite{Wilke00}.

\begin{definition}\textit{(Syntactic alternation depth)}
 For an arbitrary formula $\varphi \in L_{\mu}$, its alternation depth $\alpha(\varphi):L_{\mu}\rightarrow \mathbb{N}$ is function defined inductively:

\begin{itemize}
 \item $\alpha(\bot)=\alpha(\top)= \alpha(p) = \alpha(\neg p)= 0$,
 \item $\alpha(\varphi \circ \psi) = max\{\alpha(\varphi),\alpha(\psi)\}$ where $\circ\in\{\wedge,\vee\}$,
 \item $\alpha(\Box \varphi) = \alpha(\Diamond \varphi) = \alpha(\varphi)$,
 \item $\alpha(\mu p\varphi) = max\{1,\alpha(\psi)\} \cup \{\alpha(\nu p'\psi') + 1 |\nu p'\psi'\leq \psi,p\in free(\nu p'\psi')\}$,
\item $\alpha(\nu p\varphi) = max\{1,\alpha(\psi)\}\cup \{\alpha(\mu p'\psi') + 1 |\mu p'\psi'\leq \psi,p\in free(\mu p'\psi')\}$.
\end{itemize}
\end{definition}
Alternation depth of a formula is greater or equal to the alternation depth of any subformula. Note that, computing the alternation depth of all subformulae of an $L_{\mu}$ can be done in $\mathcal{O}((|\varphi|_{\mathcal{K}})^2 +  |\varphi|log|\varphi|) \subseteq \mathcal{O}(|\varphi|^2)$ time in bottom up fashion.
\begin{example}
Let consider a formula $\varphi = \mu p_1(\nu p_1(p_0  \wedge p_1)\vee \Diamond p_1)$, its alternation depth is $\alpha(\varphi) = 1$  and $ \psi  = \mu p_1((p_2 \wedge p_0)\vee p_1)$ and $\varphi = \nu p_2(\psi)$ it follows $\alpha(\varphi) = 2$ since $p_1 \in free(\psi)$. 
\end{example}

\subsection{Lattice Theory \& Monotone Functions}
The semantics of the $\mu$\mbox{-}Calculus is anchored in the Tarski-Knaster theorem, giving a means to do iteration basedmodel checking in an efficient manner. In order to introduce the concept of Tarski theorem some basic notions of lattice theory and fixed-point are revisited as followed:

\begin{definition}\textit{(Lattice Theory)}
 A lattice $(L,\leq)$ consists of a set $L$ and a partial order $\leq$ such that any pair of elements has greatest lower bound, the meet $\sqcap$, and a least upper bound, the join $\sqcup$, with following properties:

\begin{eqnarray}
\textit{ (associative law)} (x \sqcup y) \sqcup cup = x \sqcup ( y \sqcup z). \\
(x \sqcap y) \sqcap cup = x \sqcap ( y \sqcap z).\\
 \textit{ (commutative law)} (x \sqcup y)  = ( y \sqcup x).\\
(x \sqcap y)  = ( y \sqcap x). \\
 \textit{ (idempotency law)} (x \sqcup x)  = (x \sqcup x).\\
(x \sqcap x) = (x \sqcap x). \\
\textit{ (absorption law)} x \sqcup ( x \sqcap y) = x. \\
x \sqcap ( x \sqcap y) = x.
\end{eqnarray}
\end{definition}

\begin{example}
 Given a set $S$, the power set of $S$, (i.e., $\mathcal{P}(S)$) is $(\mathcal{P}(S),\subseteq)$ is a lattice.
\end{example}

A Lattice $(L, \leq,\sqcup, \sqcap)$ is \textit{complete} if $\forall A \subseteq L$ implies $\sqcup A$ and $\sqcap A$ are defined and there also exists a \textit{minimum element} (i.e., $\bot = \sqcap L$) and a \textit{maximum element} (i.e., $\top = \sqcup L$).

\begin{example}
 Given a set $S$, the power set of $S$, (i.e., $\mathcal{P}(S)$) is $(\mathcal{P}(S),\subseteq)$ is a lattice.
For a given set $A \subseteq \mathcal{P}(S)$ of subsets such that maximal set $\sqcup A = \bigcup_{S'\in A}S'$ and minimal set $\sqcap A = \bigcap_{S'\in A}S'$:
\begin{eqnarray}
 \sqcup A = \{\bigcup_{S'\in A}S'| A \supseteq \mathcal{P}(S')\} \\
 \sqcap A = \{\bigcap_{S'\in A}S'| A \subseteq \mathcal{P}(S')\} \\
\end{eqnarray}
\end{example}

\begin{definition}\textit{(Monotone Functions)}

A monotonic function (or monotone function) is a function which preserves the given order is formalized as followed:
\begin{itemize}
 \item $f:L \rightarrow L$ is a monotonic order preserving if
\begin{equation}
 \forall x,y \in L. x \leq y \Rightarrow f(x) \leq f(y).
\end{equation}
\item $x$ is a fix\mbox{-}point if $f(x) = x$.
\end{itemize}
 
\end{definition}

$f^0$ is an identity function and $f^{n + 1} = f^n \circ f^0$, $f$ monotonic implies that $f^n$ is also monotonic. The identity function is monotonic and composing two monotonic functions gives a monotonic function. 

\subsection{Tarski\mbox{-}Knaster Fix\mbox{-}point Theorem}
\begin{theorem}
 Let $f: L \rightarrow L$ be a monotonic function on a complete lattice $(L, \leq,\sqcup, \sqcap)$ then for $A = \{y|f(y)\leq y\}$, $x = \sqcap A$ is the \textit{least fixed point} of $f$.
\end{theorem}
\begin{proof}\textit{sketch:} 
\begin{itemize}
 \item(1) $f(x) \leq x$: $\forall y \in A$, $x\leq y$ therefore $f(x)\leq f(y) \leq (x)$. So, $f(x) = \sqcap A = x$.
 \item(2) $x \leq f(x)$: by monotonicity applied to 1, $f^2(x)\leq f(y)$ so $f(x) \in A$, and $x\leq f(x)$. 
Thus $x$ is a fixed point and because all fixed points belong to $A$, $x$ is the least fixed point. Similarly for the greatest fix\mbox{-}point with $A = \{y|f(y)\geq y\}$.
\end{itemize}
\end{proof}

\subsection{Semantics of $\mu$\mbox{-}Calculus}
The formulas of modal $\mu$\mbox{-}Calculus are interpreted in Kripke structures $\mathcal{K}$ such that for every Kripke structure $\mathcal{K}$ and every $\varphi \in L_{\mu}$, where $\kappa := \mathcal{P} \rightarrow 2^{S}$ is defined as:

\begin{itemize}
 \item $||\bot||_{\mathcal{K}} = \emptyset$,  $||\top||_{\mathcal{K}} = S$,
 \item $||p||_{\mathcal{K}} = \kappa(p)$, $||\neg p||_{\mathcal{K}} = S/\kappa(p)$,
 \item $||\varphi_1 \vee \varphi_2||_{\mathcal{K}} = ||\varphi_1||_{\mathcal{K}} \cup ||\varphi_2||_{\mathcal{K}}$,
 \item $||\varphi_1 \wedge \varphi_2||_{\mathcal{K}} = ||\varphi_1||_{\mathcal{K}} \cap ||\varphi_2||_{\mathcal{K}}$,
\item $||\Box \varphi||_{\mathcal{K}} = \{s\in S| sR \subseteq ||\varphi||_{\mathcal{K}}\}$,
\item $||\Diamond \varphi||_{\mathcal{K}} = \{s\in S| sR \cap ||\varphi||_{\mathcal{K}} \neq \emptyset\}$. 
\end{itemize}
 
We will say that pointed Kripke structure $(\mathcal{K},s_I)$ is a model of $\varphi\in L_{\mu}$, denoted by $\mathcal{K}\models \varphi$ if $s_I\in ||\varphi||_{\mathcal{K}}$. Aditionally, we write $\varphi \equiv \psi$ if for all Kripke models $(\mathcal{K},s_I)$, we have $(\mathcal{K}, s_I)\models \varphi$ iff $(\mathcal{K}, s_I)\models \psi$. 

To define the semantics of the fixed\mbox{-}point operators where $\mathcal{K}$ is a Kripke structure, $p$ is a propositional variable and $S' \subseteq S$ then $\mathcal{K}[p\mapsto S']$ denotes the Kripke structure as followed:
\begin{equation}
 \mathcal{K}[p\mapsto S'] = (S, E, \kappa[p\mapsto S']) 
\end{equation}
where $\kappa[p\mapsto S']$ is given as followed:
\begin{equation*}
\mathcal{K}[p\mapsto S']p' =
\left\{
\begin{array}{rl}
S' & \text{if } p' = p\\
\kappa(p) & \text{if } p' \neq p
\end{array} \right.
\end{equation*}

The semantics of the fixed\mbox{-}point operators is now defined as:
\begin{itemize}
\item $||\mu p\varphi||_{\mathcal{K}} = \bigcap\{S' \subseteq S|\texttt{ } ||\varphi||_{\mathcal{K}[p\mapsto S']}\subseteq S'\} $
\item $||\nu p\varphi||_{\mathcal{K}} = \bigcup \{S' \subseteq S|\texttt{ } ||\varphi||_{\mathcal{K}[p\mapsto S']}\supseteq S'\} $
\end{itemize}

$\mu z.f(z)$, the least fix\mbox{-}point of $f$ is equal to $\sqcup_i f^i(\emptyset)$, where $i$ ranges over all ordinals of cardinality at most the state space $L$; when $L$ is finite, $\mu z.f(z)$  is the union of following ascending chain $\bot \subseteq f(\bot) \subseteq f^2(\bot)\dots$

$\nu z.f(z)$ =  $\sqcap_i f^i(\top)$, where $i$ ranges over all ordinals of cardinality at most the state space $L$; when $L$ is finite, $\nu z.f(z)$  is the intersection of following descending chain $\top \supseteq f(\top) \supseteq f^2(\top)\dots$

\begin{example}
 Consider a formula $\varphi = \mu p(\Box p)$. This formula characterizes those worlds of a Kripke model where only path with finite length exists. That is, we have $s_I \in ||\varphi||_{\mathcal{K}}$ iff all paths in $\mathcal{K}$ starting in $s_I$ are finite. Lets consider this formula evaluation in detail;

\begin{proof}
"$\Leftarrow$": If all paths in $\mathcal{K}$ starting in $s_I$ are finite then $s_I \in ||\varphi||_{\mathcal{K}}$ if all paths in $\mathcal{K}$ starting in $s_I$ are finite.

Let consider $\mathcal{K} = (S,R,\lambda)$ be an arbitrary Kripke model. In order to show $S_f \subseteq ||\varphi||_{\mathcal{K}}$, Suppose $S_f \subseteq S$ denote only states(worlds) of finite length paths. Accordingly to above defined semantics of $L_{\mu}$ for least fixed point $||\Box\varphi||_{\mathcal{K}[p'\mapsto S']}$ implies $S_f \subseteq S$ holds for all $S' \subseteq S$. By unwinding $\Box$ operator $||\Box\varphi||_{\mathcal{K}[p'\mapsto S']} = \{s\in S| sR \subseteq ||\varphi||_{\mathcal{K}[p'\mapsto S']}\}$ implies $S_f\subseteq S'$ (i.e., All paths are finite).
 \end{proof}
\end{example}

\section{ Encoding $\mu$-Calculus to Alternating Tree Automaton ($L_{\mu}$ = $L_{\mathcal{A}}$)}
To define the translation of $\mu$\mbox{-}Calculus formulas into an alternating tree automaton $\mathcal{A}$, let us consider a notion of subformula property and a mapping for every subformula to a state tuple~\cite{Zappe02}~\cite{Wilke00}.

\begin{definition}
 For every formula $\varphi \in L_{\mu}$, subformula of $\varphi$ defined as follows:
\begin{itemize}
 \item $\varphi$ is subformula of $\varphi \in L_{\mu}$,
 \item $\varphi,\psi$ is subformula of $\varphi\circ \psi \in L_{\mu}$ where $\circ \in \{\wedge,\vee\}$,
 \item $\varphi$ is subformula of $\Box\varphi,\Diamond\varphi,\mu p\varphi,\nu p\varphi \in L_{\mu}$.
\end{itemize}
\end{definition}

The function which simple unwinds the formula by defining a state correspondence for every consecutive subformula:
\begin{equation}
 ||\varphi||  = \langle\varphi_0,\dots,\varphi_n\rangle \texttt{ for every } \varphi_i \in Sub(\varphi) \texttt{, } i \leq |\varphi|.
\end{equation}

Let $\varphi$ be an $L_{\mu}$ formula in normal form. The alternating tree automaton $\mathcal{A}(\varphi)$ is defined by:
\begin{equation}
 \mathcal{A} = (Q, q_I,\delta, \Omega).
\end{equation}
where
\begin{itemize}
 \item Q is the set which contains for each subformula $\psi$ of $\varphi$ (including $\varphi$ itself), a state denoted by $\langle\psi\rangle$,
\item the initial state is given by $q_I = \langle\varphi\rangle$.
\end{itemize}

The transition relation $\delta$ is defined by:
\begin{itemize}

\item $\delta(\langle\bot\rangle) = 0$,  
\item $\delta(\langle\top\rangle) = 1$,
\item $\delta(\langle p\rangle)= 
\left\{
\begin{array}{rl}
p & \textit{ if } p \in free(\varphi)\\
\langle\varphi_p\rangle & \textit{ if } p\not\in free(\varphi)
\end{array}
\right.$	
\item $\delta(\langle\neg p\rangle) = \neg p$ ,
\item  $\delta(\langle\psi_1 \wedge \psi_2\rangle) = \langle\psi_1\rangle \wedge \langle\psi_2\rangle$, $\delta(\langle\psi_1 \vee \langle\psi_2\rangle) = \langle\psi_1\rangle \vee \langle\psi_2\rangle$ ,
\item $\delta(\langle\Diamond \psi\rangle) = \Diamond \langle\psi\rangle$, 
\item $\delta(\langle\Box \psi\rangle) = \Box \langle\psi\rangle$ ,
\item $\delta(\langle\mu p \psi\rangle) = \langle\psi\rangle$, 
\item $\delta(\langle\nu p \psi\rangle) = \langle\psi\rangle$.
\end{itemize}

The priority function $\Omega:Q \rightarrow \omega$ is defined by:

\begin{equation*}
\Omega(\langle\psi\rangle) =
\left\{
\begin{array}{rl}
2 \lceil \alpha(\psi)/2\rceil - 1 & \textit{ if } \psi \in F_{\mu},\alpha(\psi) > 0\\
2 \lfloor \alpha(\psi)/2\rfloor & \textit{ if } \psi \in F_{\nu},\alpha(\psi) > 0\\
0 & \textit{otherwise}
\end{array}
\right. 
\end{equation*}

The notion of index of an alternating tree automaton coincide with the notion of alternation depth. Let $\varphi$ be an $L_{\mu}$ formula then $\alpha(\varphi) = index(\mathcal{A}(\varphi))$.

\begin{example}
 Consider $\varphi = \mu q_0(q_0 \vee q_1)$ then the transition function with defined mapping is as follows:
\begin{eqnarray}
 \delta(\langle\mu q_0(q_0 \vee q_1)\rangle) = \langle q_0 \vee q_1\rangle\\
 \delta(\langle q_0 \vee q_1\rangle) = \langle q_0\rangle \vee \langle q_1\rangle\\
 \delta(\langle q_0\rangle) = q_0.
\end{eqnarray}
\end{example}
 
A vertex $v = (\langle\psi\rangle, s)$ belongs to Player $0$ iff 
\begin{itemize}
\item $\psi = \bot$,
\item $\psi = p$, $p \in free(\varphi),$  $s \not\in \lambda(s)$,
\item $\psi = \neg p$, $p \in free(\varphi),$  $s \in \lambda(s)$,
\item $\psi = p$, $p \not\in free(\varphi)$,
\item $\psi = \eta p\psi'$ where $\eta \in \{\mu,\nu\}$,
\item $\psi = \psi_1 \vee \psi_2$ for some $\psi_1,\psi_2 \in L_{\mu}$,
\item $\psi = \Diamond \psi'$.
\end{itemize}

A vertex $v = (\langle\psi\rangle, s)$ belongs to Player $1$ iff 
\begin{itemize}
\item $\psi = \top$,
\item $\psi = p$, $p \in free(\varphi),$  $s \in \lambda(s)$,
\item $\psi = \neg p$, $p \in free(\varphi),$  $s \not\in \lambda(s)$,
\item $\psi = p$, $p \in free(\varphi)$,
\item $\psi = \psi_1 \wedge \psi_2$ for some $\psi_1,\psi_2 \in L_{\mu}$,
\item $\psi = \Box \psi'$.
\end{itemize}

In parity Game $\mathcal{G}$, the edge relation $E^{\mathcal{G}}$ is defined as:
\begin{equation*}
E^{\mathcal{G}} = 
\left\{
\begin{array}{ccc}
\{(\langle\psi'\rangle ,s)|\langle\psi'\rangle \in \delta(\langle\psi\rangle)\} & \textit{ if } \psi \neq \Diamond\psi',\Box \psi'\\
\{(\langle\psi'\rangle ,s')|\langle\psi'\rangle \in \delta(\langle\psi\rangle),s' \in sR\} & \textit{ if } \psi = \Diamond\psi',\Box \psi'
\end{array}
\right.
\end{equation*}

For $\mu$\mbox{-}formula the priority is odd and for a $\nu$\mbox{-}formula priority is even.

\begin{theorem}
 Let $\varphi$ be an arbitrary $L_{\mu}$ formula. Then $\varphi$ and $\mathcal{A}(\varphi)$ are equivalent, that is:
\begin{equation}
 ||\varphi|| = ||\mathcal{A}(\varphi)||.
\end{equation} 
\end{theorem}

\section{Proof of Correctness}

\begin{theorem}
 Let $\varphi$ be an arbitrary $L_{\mu}$ formula. Then for every pointed transition system $(\mathcal{K},s)$ the following holds:
\begin{equation}
 (\mathcal{K},s) \models \varphi \texttt{  iff } (\mathcal{K},s)\in L(\mathcal{A(\varphi)}) 
\end{equation} 
\end{theorem}
\begin{proof}
 We proceed by induction on the size of the formula $\varphi$:

\textbf{Case:} $\varphi = \top$. Clearly, every Kripke structure $(\mathcal{K},s_I)$ is a model of $\varphi$. Thus every pointed transition system is accepted by $\mathcal{A}(\varphi)$. The initial state of game $\mathcal{G}(\mathcal{A}(\varphi),\mathcal{K},s_I)$ is a $Vertex-1$ and is dead-end. Hence, every game in this play is won by Player $0$.

\textbf{Case:} $\varphi = \bot$. Then for the complement case ($\mathcal{K},s) \not\models \bot$ and from proposition 1 it follows that automata $\mathcal{A}$ does not contain any succeeding run for $\varphi = \bot$ (i.e., $(\mathcal{K},s) \not\in L(\mathcal{A(\varphi)})$.

\textbf{Case:} $\varphi = p$. Let $(\mathcal{K},s_I)\models \varphi$ if $s_I\in\kappa(p)$. Thus in $\mathcal{G}(\mathcal{A}(\varphi),\mathcal{K},s_I)$ is $vertex-1$ and a deadend as well, therefore $(\mathcal{K},s_I)\in L(\mathcal{A}(\varphi))$. Similarly, if $(\mathcal{K},s_I)\not\models \varphi$ if $\kappa(p)\in s_I$ then we have $s_I\not\in\kappa(p)$, thus $\mathcal{G}(\mathcal{A}(\varphi),\mathcal{K},s_I)$ is $vertex-0$ and a dead\mbox{-}end. Therefore, $(\mathcal{K},s_I)\not\in L(\mathcal{A}(\varphi))$.

\textbf{Case:} $\varphi = \neg p$. Similar to the previous case.

\textbf{Case:} $\varphi = \psi_1 \wedge \psi_2$, then:

\begin{align}
(\mathcal{K},s_I)\models \psi_1 \wedge \psi_2 &=  (\mathcal{K},s_I)\models \psi_1 \wedge \psi_2\\
&= s_I\in ||\psi_1||_{\mathcal{K}}\cap s_I\in ||\psi_2||_{\mathcal{K}}	& (L_{\mu}, induction)\\
&= (\mathcal{K},s_I)\models \psi_1 \texttt{ and } (\mathcal{K},s_I)\models \psi_2 & (def.)	\\
&= (\mathcal{K},s_I)\in L(\mathcal{A}(\psi_1)) \cap L(\mathcal{A}(\psi_2)) & (lemma.3\footnotemark[1])\\ % footnotes
&= (\mathcal{K},s_I)\in L(\mathcal{A}(\psi_1\wedge\psi_2)) & (lemma.3\footnotemark[1]))
\end{align}
 \footnotetext[1]{Appendix for listed proof;}
\textbf{Case} $\varphi = \psi_1 \vee \psi_2$, similar to the previous case, $lema.4\footnotemark[1])$ is used instead.

\textbf{Case} $\varphi = \Box \psi$

\begin{align}
(\mathcal{K},s_I)\models\Box\varphi &= s_I\in ||\Box\psi||_{\mathcal{K}}	& (definition)\\
 &= sR\in ||\psi||_{\mathcal{K}}	& (L_{\mu})\\
 &= \forall s' \in sR. (\mathcal{K},s')\models \psi.	& (def.)\\
 &= \forall s' \in sR. (\mathcal{K},s')\in L(\mathcal{A}(\psi)).	& (def.)\\
 &= (\mathcal{K},s')\in L(\mathcal{A}(\Box\psi)).	& (lemma.5\footnotemark[1])) 
\end{align}

\textbf{Case} $\varphi = \Diamond \psi$, similar to the previous case, $lemma.6\footnotemark[1])$ is used instead.

\textbf{Case} $\varphi = \mu p\psi$, 	
Let ($\mathcal{K},s_I)$ be a pointed Kripke structure where $\mathcal{K} =(S,R,\lambda)$, consider a monotone function  $g:2^S\rightarrow 2^S$ such that:
\begin{equation}
g(S') = \{s'\in S|\forall S'\subseteq S.(\mathcal{K}[p\mapsto S'],s')\in L(\mathcal{A}(\psi))\} 
\end{equation}

\begin{equation}
 (\mathcal{K},s)\in L(\mathcal{A}(\mu p\psi)) \texttt{ iff Player 0 wins the game } \mathcal{G} = (\mathcal{A}(\mu p\psi),\mathcal{K},s).
\end{equation}

Two further notions $\mu g$, a montone function and $S_{\mu}$, set of winning positions of Player $0$ are defined as:
\begin{equation}
 \mu g = \{S' \subseteq S| g(S')\subseteq S\}
\end{equation}

\begin{equation}
 S_{\mu} = \{s\in S' | \texttt{ There is a winning strategy f for Player 0 in } \mathcal{G}(\mathcal{A}(\mu p\nu),\mathcal{K},s)\}.
\end{equation}

Hence we have to show that $\mu g = S_{\mu}$

"$\subseteq:$" It's suffice to show that $g(S_{\mu})\subseteq S_{\mu}$. Let $s\in g(S_{\mu})$, Player $0$ has a memoryless winning strategy $f$ in game $\mathcal{G}(\mathcal{A}(\mu p\psi),\mathcal{K}[p\mapsto S_{\mu}],s)$. The
$\mathcal{G}(\mathcal{A}(\mu p\psi),\mathcal{K},s)$ has an initial vertex $(\langle\mu p\psi\rangle,s)$ that has an edge to vertex $(\langle\psi\rangle,s)$ in game $\mathcal{G}(\mathcal{A}(\mu p\psi),\mathcal{K}[p\mapsto S_{\mu}],s)$.
At first Player $0$ moves the pebble to $(\langle\psi\rangle,s)$ then by playing in accordance with his memoryless winning strategy $f$ for the game $\mathcal{G}(\mathcal{A}(\mu p\psi),\mathcal{K}[p\mapsto S_{\mu}],s)$ the play reaches a vertex of form $(\langle p\rangle,s)$ which is a dead end in this game. Since Player $0$ has played with $f$, this vertex must belonged to Player $1$, that is $s\in S_{\mu} = \kappa[p\mapsto S_{\mu}](p)$. Thus by definition of $S_{\mu}$, Player $0$ has a winning strategy in $\mathcal{G}(\mathcal{A}(\mu p\psi),\mathcal{K}[p\mapsto S_{\mu}],s)$ since he move the pebble to $(\langle\mu p \psi\rangle,s)$ and then play in accordance with his winning strategy and wins. Therefore $g(S_{\mu})\subseteq S_{\mu}$.

"$\supseteq:$" In order to prove that $S_{\mu}\subseteq S'$, now proceed as proof by contradiction. Let $s_1 \in S_{\mu}$ but $s_1\not\in S'$. Since $f$ is a winning strategy for Player $0$, so restriction of $f$ to the vertices of $\mathcal{G}(\mathcal{A}(\psi),\mathcal{K}[p\mapsto S'],s_1)$ is not a winning strategy for Player $0$. Thus we obtain a vertex $(\langle\mu p\psi\rangle,s_2)$ and finite play $\pi_1$ in $\mathcal{G}(\mathcal{A}(\psi),\mathcal{K}[p\mapsto S'],s_1)$ which is consistent with $f$ such that $(\langle p\rangle,s_2)$ is the last vertex in $\pi_1$, $s_2\in S_{\mu}$ but $s_2\not\in S'$. Inductively we obtain an infinite sequence of vertices $(\langle p\rangle,s_i)_{i\in\omega}$ in game $\mathcal{G}(\mathcal{A}(\psi),\mathcal{K}[p\mapsto S'],s_1)$ consistent with restrictions of $f$. Hence, the following play in $\mathcal{G}(\mathcal{A}(\mu p\psi),\mathcal{K},s_0)$ which is consistent with $f$ and therefore won by Player $0$:
\begin{equation}
 \pi = (\langle\mu p\psi\rangle,s_0)\pi_0(\langle\mu p\psi\rangle,s_1)\pi_1(\langle\mu p\psi\rangle,s_2)\dots
\end{equation}
 
Since $\Omega(\langle\mu p \psi\rangle)$ is the maximum priority of the automaton $\mathcal{A}(\mu p \psi)$ and it is odd thus we have a contradiction. Thereforce, $(\langle\mu p\psi\rangle,s_0)$.

\textbf{Case} $\varphi = \nu p\psi$, similar to the previous case.

\end{proof}

\section{Model Checking \& Satisfiability}
\begin{proposition}\textit(Model checking in $\mu$\mbox{-}Calculus)
 Given a finite pointed Kripke structure $(\mathcal{K},s_I)$ and an $L_{\mu}$ formula $\varphi$, determine whether $(\mathcal{K},s_I) \models \varphi$.
\end{proposition}

Let $(\mathcal{K},s_I)$ be a pointed transition system and $\mathcal{B}$ an alternating automaton. Let first define a a parity game $\mathcal{G}$ as a tuple as follows:
\begin{equation}
 \mathcal{T} = (L_0,L_1,l_I,M,\Omega)
\end{equation}
where
\begin{itemize}
 \item $L_0 = V_0$ and $L_1 = V_1$ where $V_0,V_1$ are disjoint sets of vertices in arena $\mathcal{A}\in mathcal{G}$.
\item $l_I = v_0$ where $v_0$ is the starting vertex in $\mathcal{G}$.
\item  $M \subseteq E$ where $E$ are the edge relations among vertices of the game.
\item $\Omega:(L_0 \cup L_1)\rightarrow \omega$ is a priority function over finite range $\omega$. 
\end{itemize}
Clearly, the ordered pair $(L,M)$ where $L = L_0 \cup L_1$  is a directed graph, which is denoted $G(\mathcal{T})$ and called the game graph of $\mathcal{T}$.

\begin{proposition}\textit(Model Checking Problem Reduction to Acceptance Problem)
The alternating tree automaton $\mathcal{B}$ accepts $(\mathcal{K},s_I)$ if and only if Player $0$ has a winning strategy in the parity game $\mathcal{G}(\mathcal{T}) = (\mathcal{K},\mathcal{B},s_I)$ (i.e., $\mathcal{T} = (V_0,V_1,q_I\times s_I,\Omega))$.
\end{proposition}

\begin{proof}\textit{sketch}
 Just observe that accepting runs of $\mathcal{B}$ on $(\mathcal{S_,},s_I)$ and a winning strategy trees for Player $0$ in $\mathcal{G}(\mathcal{T}) = (\mathcal{K},\mathcal{B},s_I)$ are identical.
\end{proof}

\begin{proposition}\textit(Satisfiability in $\mu$\mbox{-}Calculus)
Given an $L_{\mu}$ formula $\varphi$, determine whether there exists a pointed Kripke structure $(\mathcal{K},s_I)$ such that $\mathcal{K}\models \varphi$. 
\end{proposition}

Solving the nonemptiness problem for alternating tree automata amounts to finding a tree that is accepted. Solving the winner problem for a parity game amounts to finding a memoryless winning strategy tree.

%Let consider an alternating tree automaton $\mathcal{C} = (U, u_I,\delta^{\mathcal{C}},\Omega^{\mathcal{C}})$ and a %parity game 
%\begin{equation}
%\mathcal{T} = (2^{S\times S}\times U, R \times U,(\{s_I,s_I\},u_I),M,\Omega^{\mathcal{T}} ) 
%\end{equation}
%where $\Omega^\mathcal{T}(x,u) = \Omega^\mathcal{C}(u)$ for every $x \in 2^{S\times S}\cup R$ and $M$ is defined two %types of moves:
%\begin{itemize}
% \item let $t = $
%\end{itemize}

\begin{proposition}\textit(Satisfiability Problem Encoding to Emptiness)
The automaton $\mathcal{C}$ accepts a pointed Kripke structure $(\mathcal{K},s_I)$ if and only if Player $0$ wins the game $\mathcal{T}$.
\end{proposition}

\section{Complexity Bounds}
\begin{theorem}~\cite{Wilke00}
\begin{itemize}
 \item The Model\mbox{-}Checking problem for $\mu$\mbox{-}Calculus, is solvable in time:
\begin{equation}
 \mathcal{O}(ln(\frac{2nkn}{b})^{\lfloor b/2 \rfloor})
\end{equation}
where $k$ is the number of worlds of the Kripke structure, $l$ is the size of accessability relation, $n$ is the number of subformulas and $b$ is the alternation depth.
\item The model checking is in $UP \cap co$\mbox{-}$UP$.
\end{itemize}
\end{theorem}

\begin{theorem}~\cite{Wilke00}
\begin{itemize}
 \item The word problem for alternating tree automaton is solvable in time
\begin{equation}
 \mathcal{O}(ln(\frac{2nkn}{b})^{\lfloor b/2 \rfloor})
\end{equation}
where $k$ is the number of worlds of the Kripke structure, $l$ is the size of accessability relation, $n$ is the number of states of automaton and $b$ is the index of the automaton.
\item The acceptance is in $UP \cap co$\mbox{-}$UP$.
\end{itemize} 
\end{theorem}

\begin{proposition}~\cite{Emerson88}~\cite{Wilke00}
\begin{itemize}
 \item The non\mbox{-}emptiness problem of alternating tree automaton is in $EXPTIME$.
 \item The satisfiability of $\mu$\mbox{-}Calculus is in $EXPTIME$. 
\end{itemize}
\end{proposition}

\section{Conclusion}
In this paper, we had introduced modal $\mu$\mbox{-}Calculus and provides its equivalent translation into an alternating tree automaton which implies that every Kripke query recognizable by an alternating tree automaton can be defined by a modal $\mu$\mbox{-}Calculus formula. We proposed how alternating tree automaton, together with parity game semantics provides a reasonable reduction from modal $\mu$\mbox{-}Calculus model checking problem to the membership problem and from satisfiability problem to the emptines problem with reasonable complexity bounds.

\bibliographystyle{plain}
\bibliography{bib-refs}

\section{Appendix}
The notation $\mathcal{G}\downarrow v$ for a parity game $\mathcal{G}= ((V_0,V_1,E),\chi)$ and $v\in V_0 \cup V_1$ to denote the subgame game $\mathcal{G}[U]$, where $U$ is the set of those vertices that are reachable from $v$.
\begin{lemma}
 Let $(\mathcal{G},v)$ be an initialized parity game, $f$ a memoryless winning strategy for Player $0$ in this game and $\pi$ a play consistent with $f$. Let $v' = \pi(i)$ for some $i \in \omega$. Then the restriction of $f$ to the vertices of $\mathcal{G}\downarrow v'$ is a winning strategy for Player $0$ in the game $\mathcal{G}\downarrow v'$.
\end{lemma}
\begin{proof}
 By assumption every play in $\mathcal{G}$ starting in $v'$ only visit vertices that are in $\mathcal{G}\downarrow v'$. In order to prove that play $\pi'$  in $\mathcal{G}\downarrow v'$ consistent with $f$ starting in $v'$, let proceed with prove by contradiction. Hence, there exists a play $\pi''$ which is consistent with $f$ starting in $v'$ but not won by Player $0$. Now take the ith-prefix $p$, which is $p = \pi(0)\dots \pi(i-1)$. Since $f$ is memoryless, the concatinated play $p\pi''$ is also consistent with $f$ by adding finite prefix to infinite play. Thus, $p\pi''$ is not won by Player $0$ since adding a prefix is invariant. Therefore a contradicting the assumption $f$ being a winning strategy. Therefore, every play $\pi$ is consistent with $f$ starting in $v$ is won by Player $0$.
\end{proof}

\begin{lemma}
 Let $\psi_1,\psi_2$ be $L_{\mu}$ formulas in normal form. Then the following is true:
\begin{equation}
 L(\mathcal{A}(\psi_1 \vee \psi_2)) = L(\mathcal{A}(\psi_1)\cup L(\mathcal{A}(\psi_2)
\end{equation}
\end{lemma}
\begin{proof}
 $\subseteq:$ Let $(\mathcal{K},s_I)\in L(\mathcal{A}(\psi_1 \vee psi_2))$. By the definition of alternating tree automata, there exists a memoryless winning strategy $f$ for Player $0$ in the parity game $\mathcal{G}(\psi_1\vee\psi_2,\mathcal{K},s_I)$ having initial vertex $(\langle \psi_1,\psi_2\rangle, s_I)$. Hence, every consistent play starting in $(\langle \psi_1,\psi_2\rangle, s_I)$ has form $(\langle \psi_1,\psi_2\rangle, s_I)(\langle\psi_i\rangle,s_I)\pi$ for some $i\in\{0,1\}$ and some play $\pi$. By Lemma-1, there is a winning strategy for Player $0$ in game $\mathcal{G}(\psi_1\vee\psi_2,\mathcal{K},s_I)\downarrow(\langle\psi_i\rangle,s_I) = \mathcal{G}(\psi_i,\mathcal{K},s_I)$. Therefore, we have $\mathcal{K,s_I}\in L(\mathcal{A}(\psi_i))\subseteq L(\mathcal{A}(\psi_1))\cup L(\mathcal{A}(\psi_2))$.

$\supseteq:$ Let $(\mathcal{K},s_I)\in L(\mathcal{A}(\psi_1))\cup L(\mathcal{A}(\psi_2))$. The Player $0$ has a memoryless winning strategy $f_i$ in a parity game $\mathcal{G}(\mathcal{\psi_i},\mathcal{K},s_I)$ for some $i\in \{1,2\}$. By extending the strategy $f_i$ to a strategy $f$ of Player $0$ in $\mathcal{G}(\psi_1\vee\psi_2,\mathcal{K},s_I)$ by following:

\begin{equation}
 f(v)=
\left\{
\begin{array}{ccc}
\{
f_i(v) & \textit{if v is 0-vertex in } \mathcal{G}(\psi_i,\mathcal{K},s_I)\\
(\langle\psi_i,s_I\rangle) & \textit{if } v = (\langle\psi_1\vee\psi_2\rangle,s_I)
\end{array}
\right.
\end{equation}

$f(v)$ is a strategy for Player $0$. Hence, by definition of $f$, any play $(\langle \psi_1,\psi_2\rangle, s_I)(\langle\psi_i\rangle,s_I)\pi$ for some $i\in\{0,1\}$  and some play $\pi$ is consistent with $f$. 
Thus suffix $\langle\psi_i\rangle\pi$ is consistent with $f_i$ and therefore won by Player $0$. By lemma-1,  $(\langle \psi_1,\psi_2\rangle, s_I)(\langle\psi_i\rangle,s_I)\pi$ is won by Player $0$. Therefore, $f$ is a winning strategy for Player $0$ in $\mathcal{G}(\psi_1\vee\psi_2,\mathcal{K},s_I)$ implies $(\mathcal{K},s_I)\in L(\mathcal{A}(\psi_1\vee\psi_2))$.
\end{proof}

\begin{lemma}
 Let $\psi_1,\psi_2$ be $L_{\mu}$ formulas in normal form. Then the following is true:
\begin{equation}
 L(\mathcal{A}(\psi_1 \wedge \psi_2)) = L(\mathcal{A}(\psi_1)\cap L(\mathcal{A}(\psi_2)
\end{equation}
\end{lemma}
\begin{proof}
 Given in ~\cite{Zappe02}.
\end{proof}

\begin{lemma}
 Let $\psi$ be an $L_{\mu}$ formula in normal form and transition system $\mathcal{K}=\{S,E,\lambda)$. Then the following is true:
\begin{equation}
 L(\mathcal{A}(\Box\psi)) = \{(\mathcal{K},s_I)|\forall s'\in sR:(\mathcal{K},s')\in L(\mathcal{A}(\psi))\}
\end{equation}
\end{lemma}
\begin{proof}
 $\subseteq:$ let $(\mathcal{K},s_I\in L(\mathcal{A}(\Box \psi))$ and $\mathcal{K}=\{S,E,\lambda)$. There is a memoryless winning strategy $f$ for Player $0$ in game $\mathcal{G}(\mathcal{A}(\Box \psi),\mathcal{K},s_I)$. The initial vertex $(\langle\Box \psi\rangle,s_I)$ has for each $s'\in sR$ a successor $(\langle\psi\rangle,s')$. As $f$ is a winning strategy for Player $0$, thus for prefix $(\langle\Box \psi\rangle,s_I)(\langle\psi\rangle,s')$ is consistent with $f$ for all $s'\in sR$. By Lemma-1, there is winning strategy for Player $0$ in $\mathcal{G}(\mathcal{A}(\psi),\mathcal{K},s_I)=\mathcal{G}(\mathcal{A}(\Box\psi),\mathcal{K},s_I)\downarrow(\langle\psi\rangle,s')$ for all $s'\in sR$. Therefore, for all $s'\in sR$ we have $(\mathcal{K},s')\in L(\mathcal{A}(\psi))$.
 
$\supseteq:$ For every $s'\in sR$, there is a winning strategy in $f^{s'}:V^{s'}\times V_0^{s'}\rightarrow V^{s'}$ where $V^{s'}$ and $V_0^{s'}$ denote the set of vertices in game $\mathcal{G}(\mathcal{A}(\psi),\mathcal{K},s_I)$. Now consider the game $\mathcal{G}(\mathcal{A}(\Box\psi),\mathcal{K},s_I)$ with $V$ and $V_0$ being its vertices where its initial vertex is $(\langle\Box\psi\rangle,s_I)$ and its successor are $\{(\mathcal{\psi},s')|\forall s'\in s_IR\}$. By combining strategies $f^{s'}$ and $f$ for play $\pi$ we have:
\begin{equation}
 f(\pi)=
\left\{
\begin{array}{ccc}
\{
f^{s'}((\langle\psi\rangle,s'),\pi') & \textit{if } \pi= ((\langle\psi\rangle,s))(\langle\psi\rangle,s')\psi' \textit{ for some } s'\in sR, \pi'\in (V^{s'})^*. 
\end{array}
\right.
\end{equation}

Thus, $f^{s'}$ is a winning strategy for Player $0$ in the game $\mathcal{G}(\mathcal{A}(\psi),\mathcal{K},s') = \mathcal{G}(\mathcal{A}(\Box\psi),\mathcal{K},s)\downarrow(\langle\psi\rangle,s')$ implies that $f$ is a winning strategy for Player $0$ in the game $\mathcal{G}(\mathcal{A}(\Box\psi),\mathcal{K},s)$. Therefore, $(\mathcal{K},s_I)\in L(\mathcal{A}(\Box\psi))$.
\end{proof}

\begin{lemma}
 Let $\psi$ be an $L_{\mu}$ formula in normal form and transition system $\mathcal{K}=\{S,E,\lambda)$. Then the following is true:
\begin{equation}
 L(\mathcal{A}(\Diamond\psi)) = \{(\mathcal{K},s_I)|\exists s'\in sR:(\mathcal{K},s')\in L(\mathcal{A}(\psi))\}
\end{equation}
\end{lemma}
\begin{proof}
 Given in ~\cite{Zappe02}.
\end{proof}

\end{document}